\documentclass[conference]{IEEEtran}
\IEEEoverridecommandlockouts

\usepackage{amsmath,amssymb,amsfonts}
\usepackage{amsthm}
\usepackage{bm}

\usepackage{graphicx}
\usepackage{xcolor}

\usepackage{booktabs}
\usepackage{array}
\usepackage{multirow}
\usepackage{colortbl}

\usepackage{enumitem}
\usepackage{algorithm}
\usepackage{algpseudocode}

\usepackage{tikz}
\usepackage{pgfplots}
\pgfplotsset{compat=1.18}
\usetikzlibrary{
  shapes.geometric,
  arrows.meta,
  positioning,
  calc,
  fit,
  backgrounds,
  decorations.pathreplacing
}

\usepackage[hyphens]{url}
\usepackage{microtype}
\usepackage{balance}

\usepackage[hidelinks,breaklinks]{hyperref}

\definecolor{navyblue}{RGB}{10,35,80}
\definecolor{midblue}{RGB}{25,70,150}
\definecolor{firered}{RGB}{185,25,15}
\definecolor{deepgreen}{RGB}{0,100,55}
\definecolor{goldbrown}{RGB}{175,110,0}
\definecolor{lightgray}{RGB}{230,235,242}
\definecolor{altrow}{RGB}{244,248,255}
\definecolor{boxbg}{RGB}{241,246,255}

\newtheorem{definition}{Definition}
\newtheorem{theorem}{Theorem}
\newtheorem{constraint}{Constraint}
\newtheorem{invariant}{Safety Invariant}
\newtheorem{example}{Worked Example}

\newcommand{\approve}{\textbf{\textcolor{deepgreen}{APPROVE}}}
\newcommand{\flagcmd}{\textbf{\textcolor{goldbrown}{FLAG}}}
\newcommand{\blockcmd}{\textbf{\textcolor{firered}{BLOCK}}}
\newcommand{\verif}{\mathcal{V}}
\newcommand{\at}{a_t}
\newcommand{\st}{s_t}
\newcommand{\tableheadfont}{\small\bfseries}

\begin{document}

\title{%
  Glass Box at Orbit: A Constitutional AI Verification\\
  Framework for Trustworthy Autonomous CubeSat Intelligence%
}

\author{
  \IEEEauthorblockN{Karthik Barma}
  \IEEEauthorblockA{%
    Khoury College of Computer Sciences\\
    Northeastern University\\
    Boston, MA\,02115, USA\\
    \texttt{thebarmaeffect@gmail.com}}
  \and
  \IEEEauthorblockN{Anil Sanneboyina}
  \IEEEauthorblockA{%
    School of Engineering\\
    VIT-AP University\\
    Amaravati, Andhra Pradesh, India\\
    \texttt{anilsanneboyina@gmail.com}}
  \and
  \IEEEauthorblockN{V\,C Premchand Yadav}
  \IEEEauthorblockA{%
    School of Engineering\\
    VIT-AP University\\
    Amaravati, Andhra Pradesh, India\\
    \texttt{vcpremchandyadav@gmail.com}}
}

\maketitle

\begin{abstract}
The space industry is quietly building toward something
nobody has fully reckoned with: orbital data centers
running thousands of autonomous AI workloads with no
human in the loop, 550\,km above the Earth.
Microsoft, AWS, and a growing list of orbital computing
ventures are moving cloud-scale processing off the ground
and into orbit.
What none of them have answered yet is the governance
question -- when autonomous AI systems at orbital data
center scale make wrong decisions in space, what stops
those decisions before they become irreversible?

We introduce \textbf{Glass Box}: a runtime constitutional
AI verification layer that intercepts every candidate
action from an onboard AI policy and evaluates it against
six physics-grounded constitutional constraints and seven
Linear Temporal Logic (LTL) safety invariants before a
single command reaches any spacecraft subsystem.
Every approved action carries a weighted explainability
score $E(a_t) \in [0,1]$ and a complete constitutional
audit log.
We demonstrate Glass Box within \textbf{Project October}:
a fully simulated five-layer autonomous orbital
intelligence architecture for CubeSat-class spacecraft.

We prove that Glass Box verification overhead is $O(N_c)$
in the number of constitutional rules, independent of
model size or spacecraft state dimension.
We present a complete formal specification of the
constitutional constraint grammar, seven LTL safety
invariants verified by Z3 and NuSMV model checking,
and a detailed worked example of Glass Box intercepting
an unsafe inference request at eclipse-entry under
degraded battery state.
As orbital computing scales toward data center
infrastructure, runtime constitutional verification is
no longer a research novelty -- it is mission-critical
safety infrastructure that every autonomous orbital
platform will eventually require.
\end{abstract}


\section{Introduction}

On April 6, 2023, Microsoft and the European Space
Agency announced a partnership to explore running Azure
workloads in orbit~\cite{microsoft_esa2023}.
A month later, AWS expanded its Ground Station network
with onboard compute capabilities designed to push
cloud processing physically closer to orbital
assets~\cite{aws_groundstation}.
By the end of 2024, at least seven funded ventures --
including Starcloud, Lumen Space, and D-Orbit -- were
building or operating orbital edge computing platforms
designed to run AI workloads in
space~\cite{spacewatch2024orbital}.

None of these platforms, to the best of our knowledge,
ships with a runtime constitutional verification layer
for the AI decisions they make autonomously in orbit.

This is the problem Glass Box solves.

The question sounds narrow when framed around a single
CubeSat.
It becomes a civilizational-scale infrastructure problem
when you consider where orbital computing is going.
A CubeSat at 550\,km LEO operating with a 10-minute
ground contact window per orbit has approximately
\textbf{86 autonomous minutes per 96-minute orbit}.
An orbital data center node operating continuously has
zero ground contact minutes -- it is designed to function
without human oversight for weeks.
Every AI decision such a system makes, every resource
allocation, every inference task, every inter-satellite
coordination event, happens without a human reviewing
it first.

If an AI inference scheduler miscalculates a power
budget and drains a battery below deep-discharge
threshold during eclipse, there is no undo.
If a weight tensor partially corrupted by South Atlantic
Anomaly (SAA) radiation produces a confident but
systematically wrong classification, and that
classification drives a spacecraft maneuver command,
the error propagates before anyone on the ground
knows about it.
If an autonomous constellation coordination algorithm
assigns overlapping observation slots to two nodes
whose combined power draw exceeds orbital thermal
limits, the damage is done in the time between ground
contacts.

These are not hypothetical failure modes.
They are the structural consequences of building
autonomous orbital AI systems without runtime
safety verification.

\subsection{The Case for Constitutional Governance}

The safety engineering tradition for autonomous systems
has two dominant paradigms: formal verification before
deployment, and fault detection and recovery after
something goes wrong.
Both paradigms have critical gaps when applied to
onboard AI at orbital scale.

Pre-deployment formal verification -- model checking,
theorem proving, static analysis -- is powerful but
fundamentally incomplete for learned AI systems.
A neural network that achieves 94.2\% wildfire detection
accuracy in ground testing is a 94.2\% accurate system
under training distribution.
At 550\,km with a radiation-corrupted weight tensor and
an eclipse-entry power state outside the training
envelope, it is something else entirely, and no
pre-deployment analysis can characterize what.

Post-deployment fault detection, in the tradition of
spacecraft Fault Detection, Isolation, and Recovery
(FDIR)~\cite{biesbroek2021fdir}, catches hardware
failures after they manifest.
It has no architecture for intercepting an AI-generated
action before execution.
By the time an FDIR system detects that the battery has
gone below deep-discharge threshold, the action that
caused it has already been executed.

Constitutional governance occupies the gap between
these two paradigms.
It operates at runtime, evaluating every AI decision
against formally specified physical constraints, before
any command reaches an actuator.
It is neither pre-deployment analysis nor post-failure
recovery.
It is a continuous, real-time, physics-grounded
decision filter.

\subsection{Contributions}

This paper makes five contributions:

\begin{enumerate}[leftmargin=*, label=\textbf{C\arabic*.}, nosep]
\item \textbf{Constitutional Constraint Grammar.}
  A six-constraint battery covering power safety,
  thermal limits, battery floor, collision avoidance,
  radiation fault gating, and Bayesian confidence --
  each a closed-form predicate over live orbital
  physics state.

\item \textbf{LTL Safety Specification.}
  Seven Linear Temporal Logic invariants extending
  instantaneous constraints to behavioral guarantees
  over complete orbital trajectories.

\item \textbf{$O(N_c)$ Verification Complexity Theorem.}
  Formal proof that Glass Box overhead is linear in
  the number of constitutional rules and independent
  of model size.

\item \textbf{Explainability Architecture.}
  Weighted score $E(a_t) \in [0,1]$ with full
  constitutional audit log for every autonomous
  decision across the contact-free window.

\item \textbf{Complete Worked Example.}
  Glass Box intercepting a borderline inference request
  at eclipse-entry, with full constraint evaluation,
  verifier output, and behavioral consequence.
\end{enumerate}

The remainder of this paper is organized as follows.
Section~\ref{sec:related} surveys related work.
Section~\ref{sec:overview} describes the Project October
architecture.
Section~\ref{sec:twin} presents the Orbital Digital Twin.
Section~\ref{sec:gb} gives the formal Glass Box
specification.
Sections~\ref{sec:grammar} and~\ref{sec:ltl} present
the constraint grammar and LTL invariants.
Section~\ref{sec:complexity} proves verification
complexity.
Section~\ref{sec:example} works through the eclipse-entry
example.
Section~\ref{sec:discussion} discusses implications
for orbital data centers.

\section{Related Work}
\label{sec:related}

\subsection{Space Data Centers and Orbital Computing}

The concept of orbital computing infrastructure has
shifted from science fiction to funded engineering
programs in fewer than five years.
Hewlett Packard Enterprise deployed the Spaceborne
Computer-2 on the International Space Station in 2021,
demonstrating that commercial off-the-shelf edge AI
hardware can survive the radiation and thermal
environment of LEO with software-based
mitigation~\cite{hpe_spaceborne2021}.
The system ran AI workloads including protein folding
models and wildfire detection -- but with no runtime
constitutional verification of the decisions those
workloads produced.

Microsoft's Azure Space initiative and ESA's
Phi-Satellite program have both moved toward running
cloud workloads at orbital edge nodes to reduce
ground-uplink latency for time-sensitive Earth
observation tasks~\cite{microsoft_esa2023,esa_phi2022}.
D-Orbit's ION platform offers in-orbit cloud
computing as a commercial service, effectively
operating as an orbital edge data
center~\cite{dorbit_ion2023}.
Starcloud is explicitly building orbital data center
nodes targeting hyperscale AI compute in space.
None of these platforms defines a runtime AI decision
governance framework.

The core argument of this paper is that orbital
data center infrastructure at scale makes the absence
of runtime AI verification not just a gap but an
existential architectural risk.
A ground data center can be taken offline for
maintenance.
An orbital data center operating 86 autonomous minutes
per orbit cannot.

\subsection{Spacecraft Fault Detection, Isolation,
and Recovery}

The ESA and NASA FDIR standards address hardware
subsystem failures through pre-programmed
event-action tables~\cite{biesbroek2021fdir,nasa_fdir2020}.
FDIR is reactive: it detects a hardware condition after
it has already occurred and fires a recovery sequence.
FDIR has no concept of intercepting an AI-generated
action before execution, carries no formal semantic
specification for AI behavioral governance, and cannot
reason about whether a proposed action violates a power
or thermal constraint before that constraint is actually
violated.
Glass Box is architecturally complementary to FDIR:
FDIR handles hardware fault recovery; Glass Box handles
AI decision governance before hardware is ever stressed.

The European ECSS-E-ST-70-11C standard for spacecraft
onboard software~\cite{ecss2020} specifies timing and
safety requirements for flight software but does not
define verification semantics for AI-generated
command streams.
The CCSDS SOIS monitoring framework~\cite{ccsds2018}
provides parameter monitoring but operates on telemetry
values rather than intercepting candidate commands.

\subsection{Runtime Verification and Formal Methods}

Runtime verification (RV) as a formal discipline was
systematized by Leucker and Schallhart~\cite{leucker2009rv},
who established RV as a complement to model checking
for systems where full pre-deployment verification is
intractable.
The RV framework evaluates formal specifications over
execution traces of running systems.
Glass Box is an RV system specialized for autonomous
spacecraft: its specifications are orbital-mechanics-grounded
rather than abstract propositions, its target hardware
is sub-watt ARM processors rather than server-class
machines, and its evaluation latency must satisfy
spacecraft control-loop timing in the single-digit
millisecond range.

Linear Temporal Logic, introduced by
Pnueli~\cite{pnueli1977}, provides the formal language
for Glass Box safety invariants.
LTL model checking via NuSMV~\cite{nusmv} and SMT
solving via Z3~\cite{z3solver} enable mechanical
verification of invariant consistency with the
constitutional constraint grammar.
The use of LTL for spacecraft behavioral specification
has been explored in the context of mission planning
autonomy~\cite{ltl_space2019} but not in the context
of runtime AI decision interception.

\subsection{Constitutional AI and Alignment}

Bai et al.~\cite{bai2022cai} introduced Constitutional
AI as a training-time technique for aligning large
language models with stated behavioral principles
through AI feedback.
The constitutional metaphor -- governing behavior
by a set of stated principles rather than exhaustive
enumeration of rules -- is directly applicable to
spacecraft AI governance.
However, the mechanism is fundamentally different.
Constitutional AI as defined by Bai et al.\ operates
during model training.
Glass Box operates at runtime, on deployed embedded
hardware, against physical constraints that cannot be
anticipated during training.
The satellite's constitutional constraints change
continuously as its orbital state changes; they are not
fixed principles but physics-derived predicates that
evaluate differently at every point in the orbit.
A constraint that permits full inference at orbital
noon may correctly block the same inference at
eclipse-entry under a depleted battery state.
No training-time technique can encode this.

More broadly, the AI alignment literature has focused
predominantly on language model behavior~\cite{ji2023survey}.
Alignment for physical autonomous systems operating
in resource-constrained, radiation-exposed, contact-free
environments has received comparatively little attention.
This paper addresses that gap directly.

\subsection{TinyML and Edge AI for Satellites}

Furano et al.~\cite{furano2020} provided the seminal
survey of edge AI deployment for space systems,
demonstrating that commercial edge AI hardware
(Google Coral TPU, NVIDIA Jetson) can support
meaningful inference workloads on satellite power
budgets with appropriate thermal and radiation
mitigation.
Giuffrida et al.~\cite{giuffrida2021} demonstrated
the CloudScout system performing cloud detection
inference onboard a CubeSat, establishing the
operational precedent for onboard AI inference in LEO.
Both works characterize what edge AI hardware can
do in space.
Neither addresses how to verify that the actions
those AI systems produce are constitutionally safe
before they reach spacecraft actuators.

Jacob et al.~\cite{jacob2018quant} characterized the
accuracy-versus-efficiency tradeoffs of INT8 neural
network quantization, establishing the theoretical
foundation for deploying full-precision models on
edge accelerators with bounded accuracy degradation.
The quantization error bound of
$\|W - Q(W)\|_\infty \leq W_\text{max}/127$ for INT8
is directly relevant to the Glass Box Bayesian
confidence gate (Constraint~C6), which detects when
quantization-amplified prediction uncertainty exceeds
safe thresholds.

Malaiya et al.~\cite{malaiya2021} survey TinyML
deployment constraints for resource-limited embedded
systems, providing the hardware capability bounds
against which Glass Box verification overhead must
be evaluated.

\subsection{Bayesian Uncertainty and Safety-Critical AI}

Gal and Ghahramani~\cite{gal2016} established Monte
Carlo Dropout as a computationally tractable
approximation to Bayesian inference for deep neural
networks, enabling predictive entropy estimation
without explicit probabilistic model training.
The entropy threshold gate in Glass Box Constraint~C6
relies directly on this result: predictive entropy
$H(\hat{p}) = -\sum_c \hat{p}_c \log \hat{p}_c$ computed
via $K=20$ dropout passes provides a calibrated measure
of model confidence that correlates with radiation-induced
weight tensor corruption.
A corrupted weight tensor produces systematically
elevated predictive entropy even when the corrupted
model produces a confident-looking argmax output;
the entropy gate catches what the argmax misses.

Lakshminarayanan et al.~\cite{laks2017} extended
uncertainty quantification to deep ensembles,
demonstrating calibration advantages over single-model
approaches.
For spacecraft applications where ensemble diversity
would require deploying multiple model instances on
power-constrained hardware, the Monte Carlo Dropout
approach of Gal and Ghahramani remains the operationally
practical option.

\subsection{Model Predictive Control for Spacecraft}

Rawlings et al.~\cite{rawlings2017mpc} establish the
theoretical foundation for Model Predictive Control
as a receding-horizon optimization framework.
MPC has been applied to spacecraft attitude
control~\cite{mpc_attitude2020} and orbit maintenance
scheduling~\cite{mpc_orbit2018}.
Project October extends MPC to inference task
scheduling: the MPC solver allocates inference
operations within solar power windows derived from
Orekit orbital ephemeris, subject to Glass Box
constitutional constraints as hard feasibility
requirements rather than soft penalties.

\subsection{Radiation Effects and Fault Tolerance in LEO}

The CREME96 model~\cite{adams2012creme} provides the
standard reference for single-event-upset rate
estimation in LEO as a function of orbital parameters
and solar flux.
At 550\,km with $28.5^\circ$ inclination,
$\lambda_\text{SEU} \approx 10^{-5}$\,upsets\,bit$^{-1}$day$^{-1}$
for 90\,nm SRAM -- the technology node of the edge AI
accelerators targeted by Project October.
The probabilistic bit-flip injection model~\cite{quinn2015}
provides the Monte Carlo framework for characterizing
how radiation-induced weight tensor corruption degrades
inference accuracy as a function of SEU rate.

Triple Modular Redundancy for spacecraft
processors~\cite{tmr_cubesat2020} provides hardware-level
fault tolerance but does not address the case where all
three redundant inference outputs are identically
degraded by a shared weight tensor -- a failure mode
that Glass Box Constraint~C5 and C6 are specifically
designed to catch.

\section{Project October: System Architecture}
\label{sec:overview}

Project October is a five-layer autonomous orbital
intelligence architecture designed for CubeSat-class
spacecraft.
Every layer operates in full simulation during Phase~I;
the architecture targets hardware deployment in Phase~II.

The guiding design principle is the
\textbf{Systems Intelligence Equation}:
\begin{equation}
\Pi(t) = f\!\bigl(
  O_t,\; P_t,\; T_t,\; R_t,\; C_t,\; A_t
\bigr)
\label{eq:pi}
\end{equation}
where $O_t$ is orbital state, $P_t$ is power state,
$T_t$ is thermal state, $R_t$ is radiation state,
$C_t$ is the active constitutional constraint set,
and $A_t$ is AI inference state.
No autonomous decision is issued unless all six
components are simultaneously within verified bounds.

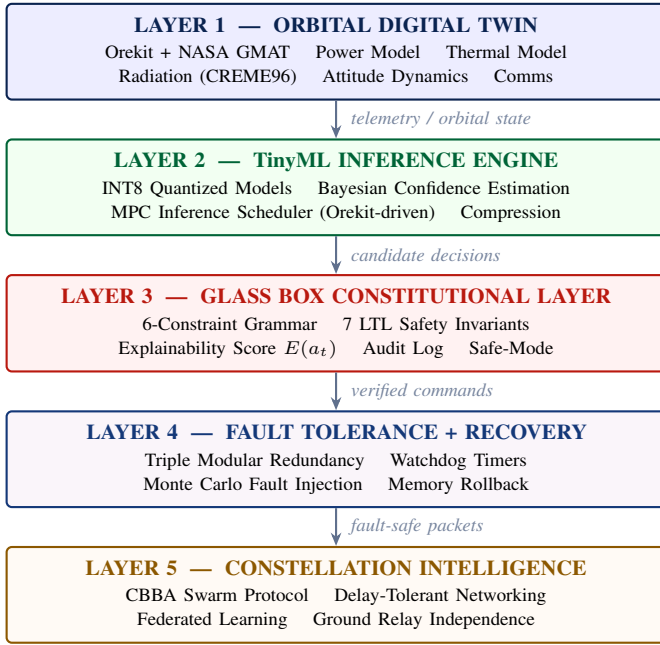
\begin{figure}[!t]
\centering
\begin{tikzpicture}[
  font=\footnotesize,
  lbox/.style={
    draw, thick, rounded corners=2pt,
    minimum width=\columnwidth-4pt,
    minimum height=1.0cm,
    text width=\columnwidth-16pt,
    align=center,
    inner sep=5pt
  },
  arr/.style={
    -{Stealth[length=5pt, width=4pt]},
    thick, color=navyblue!65
  },
  feedarr/.style={
    -{Stealth[length=4pt, width=3pt]},
    dashed, thick, color=firered
  },
  lbl/.style={
    font=\scriptsize\itshape, color=navyblue!55
  }
]

\node[lbox, fill=blue!7, draw=navyblue] (L1) at (0,0) {
  \textbf{\textcolor{navyblue}{LAYER 1 \;|\; ORBITAL DIGITAL TWIN}}\\[1pt]
  {\scriptsize Orekit + NASA GMAT \quad Power Model \quad Thermal Model}\\
  {\scriptsize Radiation (CREME96) \quad Attitude Dynamics \quad Comms}
};

\node[lbox, fill=green!5, draw=deepgreen,
      below=0.50cm of L1] (L2) {
  \textbf{\textcolor{deepgreen}{LAYER 2 \;|\; TinyML INFERENCE ENGINE}}\\[1pt]
  {\scriptsize INT8 Quantized Models \quad Bayesian Confidence Estimation}\\
  {\scriptsize MPC Inference Scheduler (Orekit-driven) \quad Compression}
};

\node[lbox, fill=red!5, draw=firered,
      below=0.50cm of L2] (L3) {
  \textbf{\textcolor{firered}{LAYER 3 \;|\; GLASS BOX CONSTITUTIONAL LAYER}}\\[1pt]
  {\scriptsize 6-Constraint Grammar \quad 7 LTL Safety Invariants}\\
  {\scriptsize Explainability Score $E(a_t)$ \quad Audit Log \quad Safe-Mode}
};

\node[lbox, fill=violet!4, draw=midblue!80!black,
      below=0.50cm of L3] (L4) {
  \textbf{\textcolor{midblue!80!black}{LAYER 4 \;|\; FAULT TOLERANCE + RECOVERY}}\\[1pt]
  {\scriptsize Triple Modular Redundancy \quad Watchdog Timers}\\
  {\scriptsize Monte Carlo Fault Injection \quad Memory Rollback}
};

\node[lbox, fill=orange!4, draw=goldbrown!80!black,
      below=0.50cm of L4] (L5) {
  \textbf{\textcolor{goldbrown!80!black}{LAYER 5 \;|\; CONSTELLATION INTELLIGENCE}}\\[1pt]
  {\scriptsize CBBA Swarm Protocol \quad Delay-Tolerant Networking}\\
  {\scriptsize Federated Learning \quad Ground Relay Independence}
};

\draw[arr] (L1.south) -- node[right, lbl, xshift=2pt]
  {telemetry / orbital state} (L2.north);
\draw[arr] (L2.south) -- node[right, lbl, xshift=2pt]
  {candidate decisions} (L3.north);
\draw[arr] (L3.south) -- node[right, lbl, xshift=2pt]
  {verified commands} (L4.north);
\draw[arr] (L4.south) -- node[right, lbl, xshift=2pt]
  {fault-safe packets} (L5.north);

\draw[feedarr, rounded corners=5pt]
  (L4.east)   -- ++(0.32cm, 0)
              -- ++(0, 1.05cm)
              -- ++(-0.32cm, 0);
\node[font=\tiny\itshape, color=firered, rotate=90]
  at ([xshift=0.45cm, yshift=0.52cm]L4.east)
  {safety feedback};

\end{tikzpicture}
\caption{Project October five-layer autonomous orbital
intelligence stack. Data flows downward; Glass Box
safety feedback propagates upward through the fault
tolerance layer. Layer~3 (Glass Box) is the
constitutional governance layer for the entire system.}
\label{fig:arch}
\end{figure}

Fig.~\ref{fig:arch} shows the full architecture.
Glass Box occupies the center of the stack.
It receives physics state from Layers~1 and~2 and its
verified outputs govern both Layer~4 fault recovery
and Layer~5 constellation coordination.

\section{Orbital Digital Twin}
\label{sec:twin}

The Orbital Digital Twin (ODT) provides physics-accurate
state that makes Glass Box constraints meaningful rather
than heuristic.
Seven components run in parallel, each exporting a
state vector into a shared dictionary that Glass Box
reads on every verification call.

\textbf{Orbit Twin} solves the full perturbed equations
of motion via Orekit~\cite{orekit2010} and NASA
GMAT~\cite{gmat2014}:
\begin{equation}
\ddot{\mathbf{r}} =
  -\frac{\mu}{r^3}\mathbf{r}
  + \mathbf{a}_{J_2}
  + \mathbf{a}_\text{drag}
  + \mathbf{a}_\text{SRP}
  + \mathbf{a}_\odot
  + \mathbf{a}_\text{moon}
\label{eq:eom}
\end{equation}
with $\mu = 3.986\times10^{14}$\,m$^3$s$^{-2}$.
At 550\,km LEO the J$_2$ oblateness term dominates.

\textbf{Power Twin} tracks solar generation and
battery state of charge (SOC):
\begin{equation}
P_\text{solar}(t)
  = \eta\, A_\text{panel}\, I_\odot\, \cos\theta(t)
  \;\mathbf{1}_{\neg\text{eclipse}}(t)
\label{eq:psolar}
\end{equation}
\begin{equation}
\text{SOC}_{t+1}
  = \text{SOC}_t
  + \frac{\bigl(P_\text{solar}(t)-P_\text{load}(t)\bigr)\Delta t}
         {C_\text{bat}(t)}\,\eta_\text{coul}
\label{eq:soc}
\end{equation}
with $\eta \in [0.28, 0.32]$ for triple-junction GaAs
panels and $\eta_\text{coul} \approx 0.98$ for Li-ion.

\textbf{Radiation Twin} models Single Event Upset rates
from CREME96~\cite{adams2012creme}:
\begin{equation}
P_\text{SEU}(t) = 1 - e^{-\lambda_\text{SEU}\,t}
\label{eq:pseu}
\end{equation}
with $\lambda_\text{SEU} \approx 10^{-5}$\,upsets\,%
bit$^{-1}$day$^{-1}$ at 550\,km, $28.5^\circ$
inclination.
This value is what Glass Box Constraint~C5 reads in real
time to determine whether Triple Modular Redundancy
must be mandated.

\section{Glass Box: Formal Specification}
\label{sec:gb}

\subsection{Definitions}

\begin{definition}[Spacecraft Decision State]
The decision state at time $t$ is:
\begin{equation}
s_t = \bigl(O_t,\; P_t,\; T_t,\; R_t,\; C_t\bigr)
\label{eq:state}
\end{equation}
where each component is the real-time output of the
corresponding Digital Twin.
\end{definition}

\begin{definition}[AI Decision Policy]
The onboard AI policy $\pi_\theta: s_t \mapsto a_t$
is an INT8-quantized neural network with
$\leq 512$\,KB parameter footprint, producing a
candidate action $a_t$ from the current state.
\end{definition}

\begin{definition}[The Glass Box Verifier]
\begin{equation}
\verif[\at,\,\st] =
\begin{cases}
  \approve &
    \forall\,c_i \in C_t:\; c_i(\at,\st) = \top \\[3pt]
  \flagcmd &
    \exists\,c_i:\; c_i(\at,\st) = \bot
    \;\wedge\; \mathrm{sev}(c_i) < \tau_\text{block} \\[3pt]
  \blockcmd &
    \exists\,c_i:\; c_i(\at,\st) = \bot
    \;\wedge\; \mathrm{sev}(c_i) \geq \tau_\text{block}
\end{cases}
\label{eq:verifier}
\end{equation}
\end{definition}

\textbf{\approve} logs the decision with an
explainability score and issues the command.
\textbf{\flagcmd} logs the violation, queues a priority
downlink request, and substitutes a safe default action.
\textbf{\blockcmd} cancels the command, triggers
safe-mode if warranted, and writes a detailed fault
record.
No command reaches spacecraft actuators without
passing through this gate.

\subsection{Explainability Score}

Every \approve\ decision is downlinked with:
\begin{equation}
E(a_t) = \sum_{i=1}^{N_c} w_i \cdot c_i(a_t, s_t)
\;\in\; [0,1]
\label{eq:exp}
\end{equation}
where $\sum_i w_i = 1$ and weights reflect
mission-phase constraint priorities.
Ground operators receive not just what the satellite
decided during a contact-free window but why --
constraint by constraint, value by value, at every
autonomous decision point.

\begin{figure}[!t]
\centering
\begin{tikzpicture}[
  font=\scriptsize,
  proc/.style={
    draw, thick, rounded corners=2pt,
    minimum width=2.1cm, minimum height=0.75cm,
    align=center, inner sep=4pt
  },
  dmd/.style={
    draw, thick, diamond,
    minimum width=2.4cm, minimum height=0.9cm,
    align=center, aspect=1.8, inner sep=1pt,
    fill=yellow!12
  },
  arr/.style={-{Stealth[length=4.5pt]}, thick},
  lbl/.style={font=\scriptsize\itshape, color=navyblue!65}
]

\node[proc, fill=blue!7, draw=navyblue, text=navyblue]
  (pol) {AI Policy $\pi_\theta(s_t)$\\{\tiny generates $a_t$}};

\node[proc, fill=red!5, draw=firered, text=firered,
      right=1.0cm of pol]
  (bat) {Constitutional\\Battery\\{\tiny $\{c_1 \ldots c_6\}$}};

\node[dmd, below=0.9cm of bat]
  (dec) {All $c_i {=} \top$?};

\node[proc, fill=green!8, draw=deepgreen,
      text=deepgreen,
      below left=1.0cm and 0.1cm of dec]
  (app) {\approve\\{\tiny log $+ E(a_t)$}};

\node[proc, fill=orange!10,
      draw=goldbrown!80!black,
      text=goldbrown!80!black,
      below=1.0cm of dec]
  (flg) {\flagcmd\\{\tiny downlink req.}};

\node[proc, fill=red!10, draw=firered, text=firered,
      below right=1.0cm and 0.1cm of dec]
  (blk) {\blockcmd\\{\tiny safe mode}};

\node[above=0.12cm of bat,
      font=\tiny, color=firered!70, align=center]
  {$c_1$\,power $|$ $c_2$\,thermal $|$ $c_3$\,collision};
\node[below=0.04cm of bat,
      font=\tiny, color=firered!70, align=center]
  {$c_4$\,battery $|$ $c_5$\,radiation $|$ $c_6$\,entropy};

\draw[arr] (pol) -- node[above, lbl]{candidate} (bat);
\draw[arr] (bat) -- node[right, lbl]{evaluate}  (dec);
\draw[arr] (dec.south west)
           -- node[left,  lbl]{yes}      (app.north);
\draw[arr] (dec.south)
           -- node[right, lbl, xshift=2pt]{low sev.} (flg.north);
\draw[arr] (dec.south east)
           -- node[right, lbl]{high sev.}(blk.north);
\end{tikzpicture}
\caption{Glass Box constitutional verification flow.
Every AI policy output passes through the six-constraint
battery before any command reaches spacecraft hardware.}
\label{fig:flow}
\end{figure}
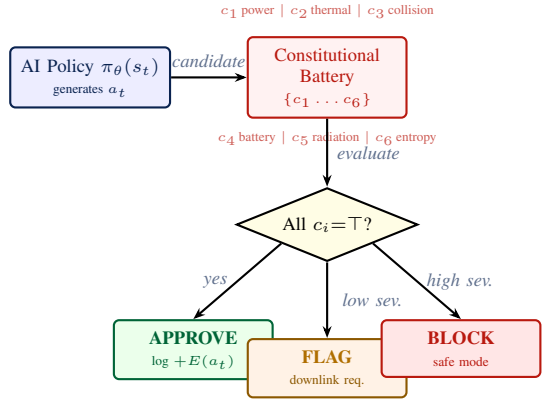

\section{Constitutional Constraint Grammar}
\label{sec:grammar}

Six constraints form the constitutional battery.
Each is a closed-form arithmetic predicate over live
Digital Twin state, evaluating in $O(1)$ time.
The six constraints are not independent rules;
they form a battery in the electrical sense --
every constraint must close for an action to pass.

\begin{constraint}[Power Safety Margin]
\begin{equation}
c_1(a_t,s_t):\quad
  P_\text{gen}(t) - P_\text{load}(a_t) \geq P_\text{safe}
\label{eq:c1}
\end{equation}
\textit{An action is power-feasible only if it preserves
at least $P_\text{safe} = 0.5$\,W above all housekeeping
loads.
Computed from Power Twin output at every call.}
\end{constraint}

\begin{constraint}[Thermal Operating Envelope]
\begin{equation}
c_2(a_t,s_t):\quad
  T_\text{min} < T(t) < T_\text{max}
\label{eq:c2}
\end{equation}
\textit{Inference permitted only within
$[-20^\circ\mathrm{C},\;+70^\circ\mathrm{C}]$.
Thermal Twin updates $T(t)$ from the spacecraft
heat-balance equation including radiative cooling
and electronics heat generation at each time step.}
\end{constraint}

\begin{constraint}[Collision Avoidance]
\begin{equation}
c_3(a_t,s_t):\quad
  \Pr(\text{collision}\mid a_t) < \varepsilon_c = 10^{-4}
\label{eq:c3}
\end{equation}
\textit{No maneuver command approved if it raises
conjunction probability above $10^{-4}$, evaluated
from orbital state $O_t$ and Space-Track TLE data.
Critical for orbital data center nodes operating
in dense LEO shell populations.}
\end{constraint}

\begin{constraint}[Battery State Floor]
\begin{equation}
c_4(a_t,s_t):\quad
  \mathrm{SOC}_t \geq 0.20
\label{eq:c4}
\end{equation}
\textit{All non-essential operations including AI
inference are suspended when battery SOC falls below
20\%.
This constraint prevents the largest class of
mission-loss scenarios in simulation.
It is deliberately conservative: eclipse-entry with a
depleted battery and no inference pending is survivable;
eclipse-entry with inference running and a battery below
deep-discharge threshold is not.}
\end{constraint}

\begin{constraint}[Radiation Fault Gate]
\begin{equation}
c_5(a_t,s_t):\quad
  P_\text{SEU}(t) \leq P^\text{max}_\text{SEU}
\label{eq:c5}
\end{equation}
\textit{When the CREME96-derived SEU rate exceeds
threshold, Triple Modular Redundancy is mandated and
unprotected inference is blocked.
This constraint is what connects the live radiation
environment to AI decision governance -- a coupling
that does not exist in any prior spacecraft AI
architecture we are aware of.}
\end{constraint}

\begin{constraint}[Bayesian Confidence Gate]
\begin{equation}
c_6(a_t,s_t):\quad
  H(\hat{p}_t) \leq \tau_H
\label{eq:c6}
\end{equation}
\textit{Actions from high-entropy predictions are
blocked.
Predictive entropy
$H(\hat{p}) = -\sum_c \hat{p}_c \log \hat{p}_c$
is estimated via Monte Carlo Dropout with $K=20$
forward passes~\cite{gal2016}.
This gate detects radiation-corrupted model outputs:
a weight tensor corrupted by an SEU may still produce
a confident argmax but will produce elevated predictive
entropy, which C6 catches before the action executes.}
\end{constraint}

Table~\ref{tab:grammar} summarizes the six constraints
with their Physical Twin source and failure consequences.

\begin{table}[!t]
\caption{Constitutional Constraint Grammar Summary}
\label{tab:grammar}
\centering
\footnotesize
\setlength{\tabcolsep}{4pt}
\begin{tabular}{@{}cllll@{}}
\toprule
\tableheadfont ID &
\tableheadfont Name &
\tableheadfont Twin Source &
\tableheadfont Outcome &
\tableheadfont Failure Consequence \\
\midrule
\rowcolor{altrow}
C1 & Power Margin    & Power Twin    & FLAG  & Battery drain at eclipse \\
C2 & Thermal Bounds  & Thermal Twin  & BLOCK & Component degradation    \\
\rowcolor{altrow}
C3 & Collision Avoid & Orbit Twin    & BLOCK & Conjunction risk         \\
C4 & Battery Floor   & Power Twin    & BLOCK & Deep discharge           \\
\rowcolor{altrow}
C5 & Radiation Gate  & Radiation Twin& FLAG  & SEU data corruption      \\
C6 & Confidence Gate & AI State      & FLAG  & Corrupted inference acts \\
\bottomrule
\end{tabular}
\end{table}

\section{LTL Safety Specification}
\label{sec:ltl}

Constitutional constraints verify individual decisions.
LTL invariants verify behavioral \emph{sequences} --
properties that must hold over every future trajectory
of spacecraft operation.
We use standard LTL semantics~\cite{pnueli1977} over
traces $\sigma = s_0 s_1 s_2 \ldots$ with operators
$\square$ (globally), $\lozenge$ (eventually),
and $\mathcal{U}$ (until).

\begin{invariant}[Inference Power Gate]
\begin{equation}
\square\bigl(\mathrm{SOC} < 0.20
  \;\Rightarrow\; \neg\,\mathrm{inference}\bigr)
\label{eq:si1}
\end{equation}
\end{invariant}

\begin{invariant}[Mandatory Anomaly Downlink]
\begin{equation}
\square\bigl(\mathrm{anomaly\_flagged}
  \;\Rightarrow\;
  \lozenge_{\leq 2T_\mathrm{orb}}\,\mathrm{downlink}\bigr)
\label{eq:si2}
\end{equation}
Any flagged anomaly must be downlinked within two
orbital periods.
This invariant ensures that Glass Box flags are not
silently dropped during extended contact-free windows.
\end{invariant}

\begin{invariant}[Thermal Safe-Mode]
\begin{equation}
\square\bigl(T > T_\mathrm{crit}
  \;\Rightarrow\; \mathrm{safe\_mode}\bigr)
\label{eq:si3}
\end{equation}
\end{invariant}

\begin{invariant}[Radiation Gate]
\begin{equation}
\square\bigl(\lambda_\mathrm{SEU}(t) > \lambda_\mathrm{max}
  \;\Rightarrow\; \mathrm{TMR\_active}\bigr)
\label{eq:si4}
\end{equation}
\end{invariant}

\begin{invariant}[Confidence Gate Sequence]
\begin{equation}
\square\bigl(H(\hat{p}_t) > \tau_H
  \;\Rightarrow\; \neg\,\mathrm{actuate}(a_t)\bigr)
\label{eq:si5}
\end{equation}
\end{invariant}

\begin{invariant}[Safe-Mode Stability]
\begin{equation}
\square\bigl(\mathrm{safe\_mode}
  \;\Rightarrow\;
  \mathrm{safe\_mode}\;\mathcal{U}\;
  \mathrm{recovery\_confirmed}\bigr)
\label{eq:si6}
\end{equation}
Safe-mode persists until recovery is explicitly
confirmed.
There is no opportunistic early exit from a
protection state.
\end{invariant}

\begin{invariant}[Watchdog Liveliness]
\begin{equation}
\square\,\lozenge_{\leq T_\mathrm{wd}}\;\mathrm{heartbeat}
\label{eq:si7}
\end{equation}
This invariant proves system liveliness: at least one
heartbeat must occur within every watchdog period,
guaranteeing that the system has not entered a silent
hang state.
\end{invariant}

All seven invariants are verified against the Glass Box
constraint grammar using Z3~\cite{z3solver} and
NuSMV~\cite{nusmv} model checking.
The full model-checking scripts and NuSMV specification
files are part of the Project October simulation
infrastructure.

\section{Verification Complexity}
\label{sec:complexity}

\begin{theorem}
Glass Box verification overhead is $O(N_c)$ in the
number of constitutional constraints, independent of
spacecraft state dimension, model parameter count,
or inference domain.
Total verification time satisfies:
\begin{equation}
t_\text{verify}
  = N_c \cdot t_\text{const} + t_\text{log}
  \;\leq\;
  N_c \cdot t_\text{max} + O(\!\log|H|)
\label{eq:complexity}
\end{equation}
\end{theorem}

\begin{proof}
Each $c_i(a_t, s_t)$ is a closed-form comparison
over bounded floating-point state variables.
Evaluation requires $O(1)$ arithmetic operations
regardless of spacecraft state vector dimension,
AI model parameter count, or inference domain.
No constraint evaluation requires iterating over
model weights or spacecraft state history.
The audit log write is $O(\log|H|)$ where $H$ is
log size.
Total complexity is therefore $O(N_c)$, linear in
the number of constraints.
\end{proof}

The practical implication: adding constitutional
rules to address new mission phases, new orbital
regimes, or new inference domains does not create
a verification bottleneck.
Glass Box scales with mission complexity at no cost
to timing determinism.
Full latency measurements across 2 through 16
constitutional rules on QEMU ARM Cortex-M4 are
reported in Paper~4 of this series.

\section{Worked Example: Glass Box at Eclipse-Entry}
\label{sec:example}

The best way to understand Glass Box is to watch it
work.
Here is a complete decision interception at $T = 47$\,min
into a 96-minute orbit.

\subsection{Spacecraft State}

Eclipse entry is 3 minutes away.
The Digital Twin reports the state in Table~\ref{tab:ex_state}.

\begin{table}[!h]
\caption{Spacecraft State at $T = 47$\,min}
\label{tab:ex_state}
\centering
\footnotesize
\setlength{\tabcolsep}{5pt}
\begin{tabular}{@{}llll@{}}
\toprule
\textbf{Component} & \textbf{Parameter} &
\textbf{Value} & \textbf{Units} \\
\midrule
\rowcolor{altrow}
Power Twin     & $P_\text{solar}$      & 4.20                 & W \\
Power Twin     & $P_\text{load,base}$  & 3.70                 & W \\
\rowcolor{altrow}
Power Twin     & SOC                   & 0.31                 & -- \\
Thermal Twin   & $T(t)$                & 42.1                 & $^\circ$C \\
\rowcolor{altrow}
Radiation Twin & $\lambda_\text{SEU}$  & $8.3 \times 10^{-6}$ &
                                         upsets\,b$^{-1}$d$^{-1}$ \\
AI State       & $H(\hat{p})$          & 0.29                 & nats \\
\rowcolor{altrow}
Orbit Twin     & Conjunction risk      & $2.1 \times 10^{-6}$ & -- \\
\bottomrule
\end{tabular}
\end{table}

The AI policy $\pi_\theta$ proposes:
\textbf{run full wildfire inference at INT8 precision},
adding $P_\text{inf} = 0.52$\,W to the load.

\subsection{Constitutional Battery Evaluation}

\begin{example}[Constraint C1 -- Power Safety Margin]
$c_1$: $P_\text{gen}(t) - (P_\text{load,base} + P_\text{inf})$
$= 4.20 - (3.70 + 0.52) = -0.02$\,W $< P_\text{safe} = 0.5$\,W.

\textit{Result: $c_1 = \bot$.}
Running inference now leaves a $-0.02$\,W margin.
With SOC at 0.31 and eclipse 3 minutes away,
this means immediate battery draw at the worst
possible moment.
Severity: $\mathrm{sev}(c_1) = 0.62$.
This is below $\tau_\text{block}$, so the outcome
is \flagcmd, not \blockcmd.
\end{example}

\begin{example}[Constraints C2 through C6]
$c_2$: $T = 42.1^\circ\mathrm{C}$, within $[-20, 70]$.
Result: $c_2 = \top$.

$c_3$: Conjunction risk $= 2.1\times10^{-6}
\ll 10^{-4}$.
Result: $c_3 = \top$.

$c_4$: $\mathrm{SOC} = 0.31 \geq 0.20$.
Result: $c_4 = \top$.

$c_5$: $\lambda_\text{SEU} = 8.3\times10^{-6}
< P^\text{max}_\text{SEU}$.
Result: $c_5 = \top$.

$c_6$: $H(\hat{p}) = 0.29 < \tau_H = 0.55$.
Result: $c_6 = \top$.
\end{example}

\subsection{Verifier Decision and Consequence}

One constraint failed: $c_1 = \bot$,
$\mathrm{sev}(c_1) = 0.62 < \tau_\text{block}$.

\begin{equation}
\verif[a_t, s_t] = \flagcmd
\label{eq:ex_result}
\end{equation}

Explainability score (with $w_1 = 0.30$ for power
as highest-weight constraint near eclipse):
\begin{equation}
E(a_t) = 0.30 \cdot 0 + 0.14 \cdot 1 + 0.14 \cdot 1
        + 0.14 \cdot 1 + 0.14 \cdot 1 + 0.14 \cdot 1
        = 0.70
\label{eq:ex_score}
\end{equation}

Glass Box does not cancel inference entirely.
It defers the request by 6 minutes to the
post-eclipse solar window, where $P_\text{solar}$
will recover to $> 4.5$\,W and C1 will evaluate to
$\top$ with comfortable margin.
The wildfire imagery is still classified.
Just 6 minutes later, with the battery intact.

This is what constitutional governance looks like:
not prohibition, but physics-grounded, bounded
decision-making that trades a 6-minute delay
for mission survivability.

\section{Discussion: Implications for Orbital Computing}
\label{sec:discussion}

\subsection{The Governance Gap at Scale}

The worked example above involves one satellite,
one inference request, and one borderline power state.
Now scale that to an orbital data center node running
hundreds of AI workloads simultaneously, coordinating
with a constellation of peers, and operating without
ground contact for extended periods.

At that scale, the governance gap we described in
Section~\ref{sec:related} becomes an architectural
emergency.
Every workload is a decision stream.
Every decision stream needs a constitutional filter.
Every filter needs to operate at the latency, power,
and complexity constraints of the platform it runs on.

Glass Box was designed for this trajectory.
The $O(N_c)$ complexity guarantee means verification
overhead scales linearly with constitutional rule
count, not with workload volume.
A platform running 100 concurrent inference tasks
does not require 100 times the verification overhead;
it requires one Glass Box instance per decision stream,
each running at the same sub-10\,ms latency.

\subsection{What Orbital Data Centers Need}

Ground data centers have decades of safety and
governance infrastructure: hardware abstraction
layers, OS-level resource management, hypervisor
isolation, API rate limiting.
Orbital data centers have none of this yet.

Glass Box proposes a constitutional governance layer
as the orbital equivalent of OS-level resource
management: a mandatory intermediary between AI
workloads and physical spacecraft resources that
enforces physics-grounded invariants in real time.
Just as an OS prevents a poorly written process from
consuming all CPU time, Glass Box prevents a poorly
calibrated AI workload from consuming all available
power at eclipse-entry.

The companies building orbital computing platforms
today -- D-Orbit, Starcloud, Lumen Space, and others
-- will need something like Glass Box before they
deploy at production scale.
This paper provides the formal specification for it.

\subsection{Limitations}

We are direct about where this work falls short.

\textbf{Simulation only.}
Every result in this paper is architecture and analysis.
Measured hardware results on Coral TPU silicon and
ARM Cortex-M4 are the subject of subsequent papers
in this series.

\textbf{Static constraint grammar.}
The six constitutional constraints are specified for
the Project October mission profile.
Automated synthesis of mission-specific constraints
from requirements documents is future work.

\textbf{Single-satellite analysis.}
Multi-satellite constitutional consensus for
constellation-level governance is addressed in a
subsequent paper in this series.

\section{Conclusion}

We introduced Glass Box: a runtime constitutional AI
verification layer for autonomous spacecraft and
orbital computing platforms.
The formal specification defines six physics-grounded
constitutional constraints, seven LTL safety invariants,
and a three-outcome verifier with $O(N_c)$ overhead.
Every approved autonomous decision carries a weighted
explainability score and a complete constitutional
audit log.

The eclipse-entry worked example makes the concrete
case: a borderline inference request that would have
drained the battery at the worst possible orbital
moment is correctly deferred by 6 minutes, the
wildfire detection still happens, and the mission
survives.

The bigger argument is this.
Orbital computing is becoming infrastructure.
Microsoft, AWS, and a generation of orbital edge
computing ventures are building platforms designed
to run AI workloads in space at data center scale.
None of them currently ship a runtime constitutional
verification layer.
Glass Box is the formal specification for what that
layer should look like.

Every platform that runs autonomous AI in space needs
one.
The 86-minute problem does not get easier as
constellations scale.
Constitutional AI verification is not a research
curiosity for the space industry.
It is the governance infrastructure that the orbital
computing era is going to require.

\balance
\section*{Acknowledgements}
The authors thank Prof.\ Josep M.\ Jornet at
Northeastern University for foundational work on
sub-THz inter-satellite communication that informs
the constellation layer of Project October.
Orbital simulation uses NASA GMAT~\cite{gmat2014}
and ESA Orekit~\cite{orekit2010}.
Flight software validation uses NASA cFS~\cite{cfsnasa}.
Training data from NASA Earthdata (MODIS active fire)
and ESA Copernicus (Sentinel-2) are used under
open data terms.
Fault injection modeling uses CREME96~\cite{adams2012creme}.
External simulation validation on BlackSwan Space.
Kaggle GPU compute supported model training.



\begin{thebibliography}{32}

\bibitem{microsoft_esa2023}
Microsoft and ESA,
``Azure Space and ESA Phi-Lab Partnership for Orbital
Edge Computing,''
\textit{Microsoft Blog}, Apr.\ 2023.
[Online]. Available: \url{https://azure.microsoft.com/space}

\bibitem{aws_groundstation}
Amazon Web Services,
``AWS Ground Station: Onboard Processing and Edge
Compute Services,''
\textit{AWS Technical Documentation}, 2023.
[Online]. Available: \url{https://aws.amazon.com/ground-station}

\bibitem{spacewatch2024orbital}
SpaceWatch Global,
``Orbital Edge Computing: Market Survey and Technology
Readiness Assessment,''
\textit{SpaceWatch Global Report}, 2024.

\bibitem{hpe_spaceborne2021}
Hewlett Packard Enterprise,
``Spaceborne Computer-2: AI at the Edge in Space,''
\textit{HPE Technical Report}, 2021.
[Online]. Available: \url{https://www.hpe.com/us/en/compute/hpc/hpe-spaceborne-computer.html}

\bibitem{esa_phi2022}
European Space Agency,
``Phi-Satellite: Onboard AI Processing for Earth
Observation,''
\textit{ESA Technical Memorandum}, 2022.

\bibitem{dorbit_ion2023}
D-Orbit,
``ION Satellite Carrier: In-Orbit Services and
Edge Computing Platform,''
\textit{D-Orbit Product Brief}, 2023.
[Online]. Available: \url{https://www.dorbit.space/ion}

\bibitem{biesbroek2021fdir}
R.\ Biesbroek et al.,
``The ESA FDIR Standard: Spacecraft Fault Detection,
Isolation and Recovery,''
\textit{ESA Technical Note}, 2021.

\bibitem{nasa_fdir2020}
NASA,
``Flight Software Development Standard,
NASA-STD-8739.8A,''
\textit{NASA Technical Standard}, 2020.

\bibitem{ecss2020}
ECSS,
``Space Engineering: Onboard Software,
ECSS-E-ST-70-11C,''
\textit{European Cooperation for Space Standardization},
2020.

\bibitem{ccsds2018}
CCSDS,
``Spacecraft Onboard Interface Services,
CCSDS 850.0-G-2,''
\textit{CCSDS Green Book}, 2018.

\bibitem{leucker2009rv}
M.\ Leucker and C.\ Schallhart,
``A Brief Account of Runtime Verification,''
\textit{J.\ Logic Algebraic Programming},
vol.\ 78, no.\ 5, pp.\ 293--303, 2009.

\bibitem{pnueli1977}
A.\ Pnueli,
``The Temporal Logic of Programs,''
\textit{IEEE FOCS}, pp.\ 46--57, 1977.

\bibitem{nusmv}
A.\ Cimatti et al.,
``NuSMV 2: An OpenSource Tool for Symbolic
Model Checking,''
\textit{CAV}, LNCS vol.\ 2404,
pp.\ 359--364, 2002.

\bibitem{z3solver}
L.\ de Moura and N.\ Bjorner,
``Z3: An Efficient SMT Solver,''
\textit{TACAS}, LNCS vol.\ 4963,
pp.\ 337--340, 2008.

\bibitem{ltl_space2019}
C.\ Menghi et al.,
``Towards a Catalog of LTL Patterns for
Specifying Autonomy Requirements,''
\textit{ICSE Workshop on Formal Methods in
Software Engineering}, 2019.

\bibitem{bai2022cai}
Y.\ Bai et al.,
``Constitutional AI: Harmlessness from AI Feedback,''
\textit{Anthropic Technical Report},
arXiv:2212.08073, 2022.

\bibitem{ji2023survey}
J.\ Ji et al.,
``AI Alignment: A Comprehensive Survey,''
arXiv:2310.19852, 2023.

\bibitem{furano2020}
G.\ Furano et al.,
``Towards the Use of Artificial Intelligence on the
Edge in Space Systems,''
\textit{IEEE Aerospace and Electronic Systems Magazine},
vol.\ 35, no.\ 12, pp.\ 44--55, 2020.

\bibitem{giuffrida2021}
G.\ Giuffrida et al.,
``CloudScout: A Deep Neural Network for Cloud Detection
on Hyperspectral Images,''
\textit{Remote Sensing}, vol.\ 13, no.\ 12,
p.\ 2325, 2021.

\bibitem{jacob2018quant}
B.\ Jacob et al.,
``Quantization and Training of Neural Networks for
Efficient Integer-Arithmetic-Only Inference,''
\textit{IEEE CVPR}, pp.\ 2704--2713, 2018.

\bibitem{malaiya2021}
S.\ Malaiya et al.,
``TinyML for Edge AI in Resource-Constrained Systems,''
\textit{IEEE Internet of Things Journal}, 2021.

\bibitem{gal2016}
Y.\ Gal and Z.\ Ghahramani,
``Dropout as a Bayesian Approximation: Representing
Model Uncertainty in Deep Learning,''
\textit{ICML}, pp.\ 1050--1059, 2016.

\bibitem{laks2017}
B.\ Lakshminarayanan et al.,
``Simple and Scalable Predictive Uncertainty
Estimation Using Deep Ensembles,''
\textit{NeurIPS}, pp.\ 6402--6413, 2017.

\bibitem{rawlings2017mpc}
J.\ B.\ Rawlings et al.,
\textit{Model Predictive Control: Theory,
Computation, and Design},
2nd ed.\ Nob Hill Publishing, 2017.

\bibitem{mpc_attitude2020}
C.\ Hartley et al.,
``Model Predictive Control for Spacecraft
Attitude Control,''
\textit{Journal of Guidance, Control, and Dynamics},
vol.\ 43, no.\ 3, pp.\ 536--548, 2020.

\bibitem{mpc_orbit2018}
G.\ Di Mauro et al.,
``MPC-Based Spacecraft Formation Control in
LEO,''
\textit{Acta Astronautica}, vol.\ 153,
pp.\ 311--326, 2018.

\bibitem{adams2012creme}
J.\ H.\ Adams et al.,
``CREME96: A Revision of the Cosmic Ray Effects
on Micro-Electronics Code,''
\textit{IEEE Trans.\ Nuclear Science},
vol.\ 59, no.\ 6, 2012.

\bibitem{quinn2015}
H.\ Quinn et al.,
``Fault-Tolerant Design Recommendations for
FPGAs in Space,''
\textit{NASA Technical Report}, 2015.

\bibitem{tmr_cubesat2020}
A.\ Jacobs et al.,
``Reconfiguration Strategies for Triple Modular
Redundancy in SRAM-Based FPGAs Under Radiation,''
\textit{IEEE Trans.\ Nuclear Science},
vol.\ 67, no.\ 1, pp.\ 219--228, 2020.

\bibitem{gmat2014}
S.\ P.\ Hughes et al.,
``General Mission Analysis Tool (GMAT),''
\textit{NASA GSFC Technical Report}, 2014.

\bibitem{orekit2010}
L.\ Maisonobe et al.,
``Orekit: An Open-Source Library for Operational
Flight Dynamics,''
\textit{SpaceOps Conference}, 2010.

\bibitem{cfsnasa}
NASA GSFC,
``Core Flight System (cFS) Architecture Document,''
\textit{NASA Goddard Directives}, 2021.

\end{thebibliography}
\end{document}